\documentclass[a4paper,11pt]{article}
\usepackage{jheppub} % for details on the use of the package, please see the JINST-author-manual
\usepackage{amsthm}
\usepackage{lineno}

\newtheorem{theorem}{Theorem}

\title{Scale-Invariant Non-Collapse from Raychaudhuri Comparison}

\author{Rohit Dhormare}
\affiliation{Dr. Babasaheb Ambedkar Marathwada University,\\
India}

\emailAdd{dhormaretheoreticalphysics.proton.me}

\abstract{We establish a scale-invariant spatial non-collapse estimate for
hypersurface-orthogonal timelike geodesic congruences. The result extends a
previous Raychaudhuri-based construction in which a comparison estimate for
the expansion was used to control the volume of domains transported along
the congruence. The additional ingredient here is a quantitative estimate of
the spatial distance distortion induced by the expansion and shear. Assuming
that the initial spatial geometry is non-collapsed up to a scale $r_0$, and
that the expansion, shear, and timelike Ricci curvature satisfy suitable
uniform bounds, we obtain, before the first blow-up time of the Raychaudhuri
comparison solution,
\begin{equation*}
\operatorname{Vol}_{h_\tau}\bigl(B_{h_\tau}(p_\tau,r)\bigr) \geq \kappa_{\mathrm{L}}(\tau)\, r^3,
\end{equation*}
where
\begin{equation*}
\kappa_{\mathrm{L}}(\tau) = \kappa_0 e^{-3K\tau} \exp\!\left(\int_0^\tau \Theta(s)\,ds\right), \qquad K = \frac{\Theta_0}{3} + \Sigma_0 .
\end{equation*}
Here $\Theta$ denotes the solution of the scalar Raychaudhuri comparison
equation and $K$ controls the spatial metric distortion. For every finite
time strictly before the comparison barrier, $\kappa_{\mathrm{L}}(\tau)$ is
positive. When the comparison solution has a finite blow-up time, the
explicit comparison coefficient approaches zero as that time is approached
from below. The comparison barrier is a boundary of the estimate and is not,
by itself, a statement that the physical congruence develops a focal point
or spacetime singularity. The result is structurally analogous to
scale-invariant non-collapse estimates in geometric flow, but is
mathematically distinct from Perelman's $\kappa$-noncollapsing theorem for
Ricci flow.}


\begin{document}
\maketitle
\flushbottom

\section{Introduction}
\label{sec:intro}

Non-collapse estimates play an important role in geometric analysis. A
central example is Perelman's $\kappa$-noncollapsing theorem for Ricci flow,
in which scale-invariant lower bounds on the volume of geodesic balls
provide quantitative control against geometric degeneration
\cite{Ham,Per1,Per2}.

A corresponding structure in Lorentzian geometry is suggested by the
Raychaudhuri equation. The equation governs the expansion of geodesic
congruences and plays a central role in the classical theory of
gravitational focusing and singularity formation \cite{Ray,HE,Wald,KarSen}.
For a timelike congruence, the expansion directly determines the evolution
of spatial volume elements orthogonal to the congruence.

In our previous work, a Raychaudhuri-based non-collapse construction was
obtained from a comparison estimate for the expansion \cite{Dhormare2026}.
The resulting estimate controls the volume of a domain transported by the
geodesic flow.

The purpose of the present Letter is to establish the next geometric step: a
scale-invariant lower bound for the volume of spatial geodesic balls. This
distinction is essential. A lower bound for the volume of a transported
domain does not, by itself, imply a lower bound of the form
\begin{equation*}
\frac{\operatorname{Vol}(B(p,r))}{r^3} \geq \kappa > 0
\end{equation*}
for geodesic balls at later times. The reason is that the spatial metric
evolves along the congruence, and therefore distances between points need
not be preserved by the flow.

We overcome this issue by combining the Raychaudhuri volume estimate with a
quantitative bound on spatial distance distortion. The basic logical
structure is
\begin{equation}
\begin{gathered}
\text{Raychaudhuri comparison} \\
\Downarrow \\
\text{transported-volume control} \\
\Downarrow \\
\text{distance distortion control} \\
\Downarrow \\
\text{spatial scale-invariant non-collapse.}
\end{gathered}
\label{eq:1}
\end{equation}

The resulting theorem is related in spirit to non-collapse estimates in
geometric flow, but it is not a Lorentzian formulation of Perelman's
theorem. In particular, the evolving spatial metric here does not satisfy
Ricci flow. Instead, the principal evolution is generated by the kinematics
of a timelike geodesic congruence through the Raychaudhuri equation.

\section{Geometric setup}
\label{sec:setup}

Let
\begin{equation*}
(M^4, g)
\end{equation*}
be a four-dimensional Lorentzian spacetime with signature $(-,+,+,+)$.

Let $u^\mu$ be a future-directed unit timelike vector field satisfying
\begin{equation}
u^\mu u_\mu = -1, \qquad u^\nu \nabla_\nu u^\mu = 0.
\label{eq:2}
\end{equation}
Thus the integral curves of $u^\mu$ are timelike geodesics.

We assume that the congruence is hypersurface orthogonal,
\begin{equation}
\omega_{\mu\nu} = 0,
\label{eq:3}
\end{equation}
so that locally the congruence defines a family of spacelike hypersurfaces
$\Sigma_\tau$.

The induced spatial metric is
\begin{equation}
h_{\mu\nu} = g_{\mu\nu} + u_\mu u_\nu.
\label{eq:4}
\end{equation}

The deformation tensor of the congruence is
\begin{equation}
B_{\mu\nu} = h_\mu{}^\alpha h_\nu{}^\beta \nabla_\alpha u_\beta.
\label{eq:5}
\end{equation}

Since the vorticity vanishes, its decomposition is
\begin{equation}
B_{\mu\nu} = \frac{1}{3}\theta h_{\mu\nu} + \sigma_{\mu\nu},
\label{eq:6}
\end{equation}
where
\begin{equation}
\theta = \nabla_\mu u^\mu
\label{eq:7}
\end{equation}
is the expansion and
\begin{equation}
\sigma^2 \equiv \sigma_{\mu\nu}\sigma^{\mu\nu} \geq 0
\label{eq:8}
\end{equation}
is the squared shear norm.

The factor $1/3$ reflects the fact that the spatial hypersurfaces are
three-dimensional.

\section{Raychaudhuri comparison}
\label{sec:comparison}

For the hypersurface-orthogonal timelike geodesic congruence, the
Raychaudhuri equation is
\begin{equation}
\dot\theta = -\frac{1}{3}\theta^2 - \sigma^2 - R_{\mu\nu}u^\mu u^\nu,
\label{eq:9}
\end{equation}
where
\begin{equation}
\dot\theta = u^\mu \nabla_\mu \theta.
\label{eq:10}
\end{equation}

Suppose that, in the spacetime region under consideration,
\begin{equation}
R_{\mu\nu}u^\mu u^\nu \leq R_0
\label{eq:11}
\end{equation}
and
\begin{equation}
\sigma^2 \leq \Sigma_0^2.
\label{eq:12}
\end{equation}

Define
\begin{equation}
A = R_0 + \Sigma_0^2.
\label{eq:13}
\end{equation}

Then eq.~\eqref{eq:9} implies
\begin{equation}
\dot\theta \geq -\frac{1}{3}\theta^2 - A.
\label{eq:14}
\end{equation}

For each geodesic, let $\Theta(\tau)$ denote the solution of the scalar
comparison equation
\begin{equation}
\dot\Theta = -\frac{1}{3}\Theta^2 - A, \qquad \Theta(0) = \theta_0,
\label{eq:15}
\end{equation}
where $\theta_0$ is the initial expansion along the geodesic under
consideration.

The scalar comparison principle gives
\begin{equation}
\theta(\tau) \geq \Theta(\tau)
\label{eq:16}
\end{equation}
on every interval on which both solutions exist.\footnote{The comparison
statement is understood pointwise along each integral curve of $u^\mu$. If
the initial expansion $\theta_0$ varies over $\Sigma_0$, then the comparison
solution $\Theta(\tau,p)$ and the corresponding comparison barrier are
likewise understood geodesic-by-geodesic. A single global comparison time
requires a uniform choice of initial and curvature bounds.}

For the case
\begin{equation}
A > 0,
\label{eq:17}
\end{equation}
define
\begin{equation}
\lambda = \sqrt{\frac{A}{3}}.
\label{eq:18}
\end{equation}

The solution can be written as
\begin{equation}
\Theta(\tau) = -\sqrt{3A}\,\tan(\lambda\tau + \alpha),
\label{eq:19}
\end{equation}
where
\begin{equation}
\tan\alpha = -\frac{\theta_0}{\sqrt{3A}}.
\label{eq:20}
\end{equation}

The first positive blow-up time of this comparison solution is
\begin{equation}
\tau_c = \sqrt{\frac{3}{A}}\left(\frac{\pi}{2} - \alpha\right).
\label{eq:21}
\end{equation}
\footnote{The quantity $\tau_c$ is the first blow-up time of the auxiliary
comparison solution. It is therefore a comparison barrier, not automatically
a focal time, conjugate-point time, or curvature singularity of the physical
spacetime. The actual expansion $\theta(\tau)$ may remain finite beyond
$\tau_c$ because the inequality $\theta \geq \Theta$ ceases to provide a
useful finite lower comparison after the auxiliary solution diverges.}

For $A \leq 0$, the comparison equation is still well defined, but its
explicit form is different. Since the main theorem below does not require
an explicit trigonometric expression, we formulate the theorem in terms of
the maximal interval on which $\Theta$ remains finite.

\section{Volume evolution}
\label{sec:volume}

The induced metric evolves along the geodesic flow according to
\begin{equation}
\mathcal{L}_u h_{\mu\nu} = 2B_{\mu\nu} = \frac{2}{3}\theta h_{\mu\nu} + 2\sigma_{\mu\nu}.
\label{eq:22}
\end{equation}

Taking the trace gives the evolution of the spatial volume element:
\begin{equation}
\frac{d}{d\tau} d\mu_\tau = \theta\, d\mu_\tau.
\label{eq:23}
\end{equation}

Let
\begin{equation*}
F_\tau : \Sigma_0 \longrightarrow \Sigma_\tau
\end{equation*}
denote the local flow map generated by $u^\mu$. More precisely, the volume
relation should be understood after pulling the later volume form back to
$\Sigma_0$:
\begin{equation}
F_\tau^* d\mu_\tau = \exp\!\left(\int_0^\tau \theta(s)\,ds\right) d\mu_0.
\label{eq:24}
\end{equation}

Using eq.~\eqref{eq:16}, one obtains
\begin{equation}
F_\tau^* d\mu_\tau \geq \exp\!\left(\int_0^\tau \Theta(s)\,ds\right) d\mu_0.
\label{eq:25}
\end{equation}

Consequently, for every measurable domain $D_0 \subset \Sigma_0$ on which
the flow is defined,
\begin{equation}
\operatorname{Vol}_\tau(F_\tau(D_0)) \geq \exp\!\left(\int_0^\tau \Theta(s)\,ds\right) \operatorname{Vol}_0(D_0).
\label{eq:26}
\end{equation}

This is the transported-domain estimate that underlies the present
construction.

\section{Spatial distance distortion}
\label{sec:distortion}

The transported-volume estimate does not by itself give a scale-invariant
estimate for spatial geodesic balls because
\begin{equation}
F_\tau(B_0(p,r)) \neq B_\tau(F_\tau(p), r)
\label{eq:27}
\end{equation}
in general.

We therefore control the deformation of the spatial metric. Assume that on
the region swept out by the flow,
\begin{equation}
|\theta| \leq \Theta_0, \qquad |\sigma|_h \leq \Sigma_0,
\label{eq:28}
\end{equation}
where
\begin{equation}
|\sigma|_h = \sqrt{\sigma_{\mu\nu}\sigma^{\mu\nu}}.
\label{eq:29}
\end{equation}

Define
\begin{equation}
K = \frac{\Theta_0}{3} + \Sigma_0.
\label{eq:30}
\end{equation}

For every spatial vector $X$, eq.~\eqref{eq:22} gives
\begin{equation}
|B(X,X)| \leq \left(\frac{|\theta|}{3} + |\sigma|_h\right) h(X,X) \leq K\, h(X,X).
\label{eq:31}
\end{equation}

Hence
\begin{equation}
-2K h \leq \mathcal{L}_u h \leq 2K h
\label{eq:32}
\end{equation}
as quadratic forms.

If $X(\tau)$ is obtained by pushing forward an initial spatial vector $X_0$
under the flow, then
\begin{equation}
\left|\frac{d}{d\tau}\log |X(\tau)|_{h_\tau}\right| \leq K.
\label{eq:33}
\end{equation}

Gronwall's inequality therefore gives
\begin{equation}
e^{-K\tau} |X_0|_{h_0} \leq |dF_\tau(X_0)|_{h_\tau} \leq e^{K\tau} |X_0|_{h_0}.
\label{eq:34}
\end{equation}

Let $d_0$ and $d_\tau$ denote the Riemannian distance functions associated
with $h_0$ and $h_\tau$, respectively. Applying eq.~\eqref{eq:34} to spatial
curves and taking the infimum over such curves yields
\begin{equation}
e^{-K\tau} d_0(p,q) \leq d_\tau(F_\tau p, F_\tau q) \leq e^{K\tau} d_0(p,q).
\label{eq:35}
\end{equation}

In particular,
\begin{equation}
d_0(p,q) \leq e^{K\tau} r \implies d_\tau(p_\tau, q_\tau) \leq r.
\label{eq:36}
\end{equation}

Therefore,
\begin{equation}
F_\tau\bigl[B_0(p, e^{-K\tau} r)\bigr] \subset B_\tau(p_\tau, r).
\label{eq:37}
\end{equation}
\footnote{The distance estimate is a local statement on the region where
the flow map is defined and the bounds in eq.~\eqref{eq:28} hold. No global
completeness assumption is required for the local inclusion of balls.}

\section{Main theorem}
\label{sec:theorem}

\begin{theorem}[Raychaudhuri-induced spatial non-collapse]
Let $(M^4, g)$ admit a smooth hypersurface-orthogonal timelike geodesic
congruence with spacelike hypersurfaces $\Sigma_\tau$. Suppose that the
initial spatial geometry satisfies the scale-invariant lower volume bound
\begin{equation}
\operatorname{Vol}_0(B_0(p,r)) \geq \kappa_0 r^3, \qquad 0 < r \leq r_0,
\label{eq:38}
\end{equation}
for some $\kappa_0 > 0$.

Assume that along the region swept out by the congruence,
\begin{equation}
R_{\mu\nu}u^\mu u^\nu \leq R_0, \qquad \sigma^2 \leq \Sigma_0^2,
\label{eq:39}
\end{equation}
and
\begin{equation}
|\theta| \leq \Theta_0, \qquad |\sigma|_h \leq \Sigma_0.
\label{eq:40}
\end{equation}

Let $\Theta$ solve
\begin{equation}
\dot\Theta = -\frac{1}{3}\Theta^2 - (R_0 + \Sigma_0^2), \qquad \Theta(0) = \theta_0,
\label{eq:41}
\end{equation}
and let $\tau_c$ denote the first finite blow-up time of the comparison
solution, if such a time exists.

Define
\begin{equation}
K = \frac{\Theta_0}{3} + \Sigma_0.
\label{eq:42}
\end{equation}

Then, for every
\begin{equation}
0 \leq \tau < \tau_c
\label{eq:43}
\end{equation}
and every radius satisfying
\begin{equation}
0 < r \leq r_0 e^{K\tau},
\label{eq:44}
\end{equation}
one has
\begin{equation}
\operatorname{Vol}_\tau(B_\tau(p_\tau, r)) \geq \kappa_{\mathrm{L}}(\tau) r^3,
\label{eq:45}
\end{equation}
where
\begin{equation}
\kappa_{\mathrm{L}}(\tau) = \kappa_0 e^{-3K\tau} \exp\!\left[\int_0^\tau \Theta(s)\,ds\right].
\label{eq:46}
\end{equation}

In particular,
\begin{equation}
\kappa_{\mathrm{L}}(\tau) > 0
\label{eq:47}
\end{equation}
for every finite $\tau < \tau_c$.
\end{theorem}

\section{Proof}
\label{sec:proof}

\begin{proof}
Fix $p \in \Sigma_0$, $\tau < \tau_c$, and a radius $r$ satisfying
eq.~\eqref{eq:44}. Define
\begin{equation}
\rho = e^{-K\tau} r.
\label{eq:48}
\end{equation}

The radius restriction gives
\begin{equation}
\rho \leq r_0.
\label{eq:49}
\end{equation}

Therefore, the initial non-collapse assumption implies
\begin{equation}
\operatorname{Vol}_0(B_0(p,\rho)) \geq \kappa_0 \rho^3 = \kappa_0 e^{-3K\tau} r^3.
\label{eq:50}
\end{equation}

By the distance distortion estimate, eq.~\eqref{eq:37},
\begin{equation}
F_\tau(B_0(p,\rho)) \subset B_\tau(p_\tau, r).
\label{eq:51}
\end{equation}

Consequently,
\begin{align}
\operatorname{Vol}_\tau(B_\tau(p_\tau, r)) &\geq \operatorname{Vol}_\tau(F_\tau(B_0(p,\rho))) \notag \\
&\geq \exp\!\left(\int_0^\tau \Theta(s)\,ds\right) \operatorname{Vol}_0(B_0(p,\rho)).
\label{eq:52}
\end{align}

Using eq.~\eqref{eq:50}, we obtain
\begin{equation}
\operatorname{Vol}_\tau(B_\tau(p_\tau, r)) \geq \kappa_0 e^{-3K\tau} \exp\!\left(\int_0^\tau \Theta(s)\,ds\right) r^3.
\label{eq:53}
\end{equation}

This is precisely eq.~\eqref{eq:45} with
\begin{equation*}
\kappa_{\mathrm{L}}(\tau) = \kappa_0 e^{-3K\tau} \exp\!\left(\int_0^\tau \Theta(s)\,ds\right).
\end{equation*}

Since $\Theta$ is finite for every $\tau < \tau_c$, the exponential factor
is strictly positive. Hence $\kappa_{\mathrm{L}}(\tau) > 0$.
\end{proof}

\section{Explicit comparison coefficient}
\label{sec:explicit}

For $A > 0$, the comparison solution is given by eq.~\eqref{eq:19}.
Integrating gives
\begin{align}
\int_0^\tau \Theta(s)\,ds &= -\sqrt{3A} \int_0^\tau \tan(\lambda s + \alpha)\,ds \notag \\
&= 3\ln\!\left[\frac{\cos(\lambda\tau + \alpha)}{\cos\alpha}\right].
\label{eq:54}
\end{align}

Thus the non-collapse coefficient becomes
\begin{equation}
\kappa_{\mathrm{L}}(\tau) = \kappa_0 e^{-3K\tau} \left[\frac{\cos(\lambda\tau + \alpha)}{\cos\alpha}\right]^3.
\label{eq:55}
\end{equation}

For $0 \leq \tau < \tau_c$, the cosine ratio is positive, and therefore
$\kappa_{\mathrm{L}}(\tau) > 0$.

As $\tau \to \tau_c^-$, one has $\cos(\lambda\tau + \alpha) \to 0$, so that
\begin{equation}
\kappa_{\mathrm{L}}(\tau) \to 0.
\label{eq:56}
\end{equation}

This limiting behavior concerns the comparison coefficient. It should not
be interpreted as proving that the physical geodesic balls necessarily
collapse at $\tau_c$. Rather, the comparison estimate loses its positive
quantitative lower bound as the auxiliary comparison solution diverges.

\section{Scale-invariant interpretation}
\label{sec:scale-invariant}

Define the spatial volume ratio
\begin{equation}
V_{\mathrm{L}}(p,\tau,r) = \frac{\operatorname{Vol}_\tau(B_\tau(p,r))}{r^3}.
\label{eq:57}
\end{equation}

The theorem gives
\begin{equation}
V_{\mathrm{L}}(p,\tau,r) \geq \kappa_{\mathrm{L}}(\tau) > 0
\label{eq:58}
\end{equation}
for all admissible radii before the comparison barrier.

The quantity $V_{\mathrm{L}}$ is dimensionless. Thus the theorem is
genuinely scale-invariant in the spatial dimension three.

The estimate can be decomposed as
\begin{equation}
\kappa_{\mathrm{L}}(\tau) = \underbrace{\kappa_0}_{\text{initial non-collapse}} \cdot \underbrace{e^{-3K\tau}}_{\text{distance distortion}} \cdot \underbrace{e^{\int_0^\tau \Theta(s)\,ds}}_{\text{Raychaudhuri volume evolution}}.
\label{eq:59}
\end{equation}

This decomposition is useful because the three factors have different
geometric origins. The first is an initial spatial assumption, the second
measures the cost of converting transported balls into geodesic balls, and
the third is the dynamical volume contribution coming from Raychaudhuri
comparison.

\section{Relation to geometric-flow non-collapse}
\label{sec:relation}

The present result has a structural resemblance to scale-invariant
non-collapse estimates appearing in geometric analysis \cite{Ham,Per1,Per2}.
Nevertheless, it should not be identified with Perelman's
$\kappa$-noncollapsing theorem.

Perelman's theorem concerns Riemannian Ricci flow,
\begin{equation}
\partial_t g = -2\operatorname{Ric},
\label{eq:60}
\end{equation}
and is deeply connected with the entropy and monotonicity structure of
Ricci flow \cite{Per1,Per2}.

In the present setting, the spatial metric satisfies instead
\begin{equation}
\mathcal{L}_u h = \frac{2}{3}\theta h + 2\sigma,
\label{eq:61}
\end{equation}
while the scalar expansion satisfies
\begin{equation}
\dot\theta = -\frac{1}{3}\theta^2 - \sigma^2 - R_{\mu\nu}u^\mu u^\nu.
\label{eq:62}
\end{equation}

Thus the underlying geometric evolutions are different.

The analogy lies specifically in the appearance of a dimensionless lower
volume ratio:
\begin{equation}
\frac{\operatorname{Vol}(B_r)}{r^3}.
\label{eq:63}
\end{equation}

Accordingly, the present result should be described as a
Raychaudhuri-induced Lorentzian spatial non-collapse estimate, rather than
as a Lorentzian version of Perelman's theorem.

\section{Connection with the previous Raychaudhuri construction}
\label{sec:connection}

The previous construction established a comparison estimate of the form
\begin{equation}
\theta \geq \Theta
\label{eq:64}
\end{equation}
and consequently the transported-domain inequality
\begin{equation}
\operatorname{Vol}_\tau(F_\tau(D)) \geq e^{\int_0^\tau \Theta\,ds} \operatorname{Vol}_0(D).
\label{eq:65}
\end{equation}

The present work adds the spatial metric estimate
\begin{equation}
e^{-K\tau} d_0 \leq d_\tau \leq e^{K\tau} d_0.
\label{eq:66}
\end{equation}

This allows one to choose an initial ball of radius $e^{-K\tau} r$ whose
image lies inside the radius-$r$ geodesic ball at time $\tau$.

The resulting logical structure is therefore
\begin{equation}
\text{transported-volume estimate} + \text{distance distortion estimate} \implies \text{spatial scale-invariant non-collapse.}
\label{eq:67}
\end{equation}

The present result is thus a geometric extension of the previous
Raychaudhuri-based volume construction rather than an independent
replacement of it.

\section{Special cases}
\label{sec:special}

\subsection{Minkowski spacetime}
\label{sec:minkowski}

Consider a static timelike congruence in Minkowski spacetime. Then
\begin{equation}
R_{\mu\nu} = 0, \qquad \theta = 0, \qquad \sigma_{\mu\nu} = 0.
\label{eq:68}
\end{equation}

Consequently,
\begin{equation}
K = 0
\label{eq:69}
\end{equation}
and
\begin{equation}
\kappa_{\mathrm{L}}(\tau) = \kappa_0.
\label{eq:70}
\end{equation}

Thus the spatial volume ratio is preserved, as expected.

\subsection{Shear-free congruence}
\label{sec:shear-free}

If
\begin{equation}
\sigma_{\mu\nu} = 0,
\label{eq:71}
\end{equation}
then
\begin{equation}
K = \frac{\Theta_0}{3}.
\label{eq:72}
\end{equation}

The non-collapse coefficient becomes
\begin{equation}
\kappa_{\mathrm{L}}(\tau) = \kappa_0 e^{-\Theta_0 \tau} \exp\!\left(\int_0^\tau \Theta(s)\,ds\right).
\label{eq:73}
\end{equation}

For an isotropic spatial evolution
\begin{equation}
h_\tau = a^2(\tau) h_0,
\label{eq:74}
\end{equation}
the expansion is
\begin{equation}
\theta = 3\frac{\dot a}{a},
\label{eq:75}
\end{equation}
and therefore
\begin{equation}
d\mu_\tau = a^3(\tau)\, d\mu_0,
\label{eq:76}
\end{equation}
which is consistent with the general volume evolution equation.

\section{Discussion}
\label{sec:discussion}

The principal result of this Letter is the finite-time estimate
\begin{equation}
\frac{\operatorname{Vol}_\tau(B_\tau(p,r))}{r^3} \geq \kappa_{\mathrm{L}}(\tau) > 0, \qquad \tau < \tau_c.
\label{eq:83}
\end{equation}

The result combines two distinct consequences of the congruence kinematics.

First, Raychaudhuri comparison controls the evolution of transported
volume:
\begin{equation}
\operatorname{Vol}_\tau(F_\tau(D)) \geq e^{\int_0^\tau \Theta\, ds} \operatorname{Vol}_0(D).
\label{eq:84}
\end{equation}

Second, bounds on expansion and shear control the distortion of spatial
distances:
\begin{equation}
e^{-K\tau} d_0 \leq d_\tau \leq e^{K\tau} d_0.
\label{eq:85}
\end{equation}

Only after combining these two ingredients can one obtain a scale-invariant
lower bound for geodesic balls.

The comparison barrier has a precise but limited role. When the comparison
solution has a finite blow-up time,
\begin{equation}
\tau \to \tau_c^-
\label{eq:86}
\end{equation}
implies
\begin{equation}
\kappa_{\mathrm{L}}(\tau) \to 0
\label{eq:87}
\end{equation}
for the explicit comparison coefficient. This means that the present
lower-bound mechanism becomes degenerate at the comparison barrier. It does
not establish that the physical spatial geometry itself must collapse at
that time.

The theorem also makes clear that initial non-collapse is an essential
hypothesis. The curvature and kinematical bounds used here do not, by
themselves, generate a universal positive $\kappa_{\mathrm{L}}$ from
arbitrary initial data.

A stronger result would require an independent mechanism relating a
dynamical quantity to the initial non-collapse scale, and this remains an
open direction for future work.

The appropriate interpretation of the present result is therefore more
modest and, at the same time, more robust: Raychaudhuri focusing, together
with controlled spatial metric distortion, yields a quantitative pre-barrier
scale-invariant non-collapse estimate.

\section{Conclusion}
\label{sec:conclusion}

We have established a scale-invariant spatial non-collapse estimate for
hypersurface-orthogonal timelike geodesic congruences.

The starting point is the Raychaudhuri comparison inequality
\begin{equation}
\dot\theta \geq -\frac{1}{3}\theta^2 - A.
\label{eq:89}
\end{equation}

This yields a lower comparison solution $\Theta$ and consequently a lower
bound for the volume of domains transported along the geodesic flow.

The second ingredient is the spatial metric evolution
\begin{equation}
\mathcal{L}_u h = \frac{2}{3}\theta h + 2\sigma,
\label{eq:90}
\end{equation}
which, under uniform bounds on $\theta$ and $\sigma$, gives
\begin{equation}
e^{-K\tau} d_0 \leq d_\tau \leq e^{K\tau} d_0.
\label{eq:91}
\end{equation}

Together these estimates imply
\begin{equation}
\operatorname{Vol}_\tau(B_\tau(p,r)) \geq \kappa_0 e^{-3K\tau} e^{\int_0^\tau \Theta(s)\,ds} r^3.
\label{eq:92}
\end{equation}

Hence
\begin{equation}
\kappa_{\mathrm{L}}(\tau) > 0
\label{eq:93}
\end{equation}
throughout the finite pre-barrier regime.

The result is structurally reminiscent of non-collapse estimates in
geometric flow, but it is not an extension of Perelman's theorem: the
present evolution is generated by the Raychaudhuri equation and the
deformation of the spatial metric along a timelike congruence.

% Bibliography

\end{document}